\documentclass[conference]{IEEEtran}

\usepackage[utf8]{inputenc} 
\usepackage[T1]{fontenc}
\usepackage{amssymb}
\usepackage{url}
\usepackage{ifthen}
\usepackage {tikz}
\usepackage{caption}
\usepackage{cite}
\usepackage{amsthm}
\usepackage{algorithmic}
\usepackage{graphicx}
\usepackage{textcomp}
\usepackage{xfrac}
\usepackage{xcolor}
\usepackage{comment}
\usepackage{float}
\usepackage{bbm}
\usepackage{enumitem}
\usepackage{multirow}
\newcounter{hints}
\renewcommand{\thehints}{\roman{hints}}
\newcommand{\hintedrel}[2][]{%
  \refstepcounter{hints}%
  \if\relax\detokenize{#1}\relax\else\label{#1}\fi
  \mathrel{\overset{\mathrm{(\thehints)}}{\vphantom{\le}{#2}}}%
}

\usepackage{comment}
\usepackage{blindtext}
\usepackage{balance}

\usetikzlibrary {positioning}
\usepackage{tkz-euclide}

\usepackage[cmex10]{amsmath} 

\newcommand\NB[1][0.3]{N\kern-#1em\textcolor{red}{B}}

\DeclareMathOperator*{\given}{\vert}
\newcommand{\E}[1]{\mathsf{E}\left[#1\right]}
\newcommand{\Var}[1]{\mathsf{Var}\left(#1\right)}

\newcommand{\mmse}[1]{\mathsf{mmse}\left(#1\right)}
\newcommand{\loge}[1]{\log\left(#1\right)}
\newcommand{\e}[1]{e^{#1}}

\newtheoremstyle{note}
  {}
  {}
  {\itshape}
  {}
  {\itshape}
  {}
  {0em}
  {\thmname{#1}\thmnumber{ #2}}

\theoremstyle{note} 
\newtheorem{theorem}{Theorem}
\newtheorem{corollary}[theorem]{Corollary}
\newtheorem{lemma}[theorem]{Lemma}
\newtheorem{proposition}[theorem]{Proposition}
\newtheorem{remark}[theorem]{Remark}

\newcounter{relctr} 

\begin{document}
\title{ The Price of Distributed:\\ Rate Loss in the CEO Problem} 



 \author{%
   \IEEEauthorblockN{Arda Atalik\IEEEauthorrefmark{1},
                      Alper Köse\IEEEauthorrefmark{2}
                     and Michael Gastpar\IEEEauthorrefmark{3},
                    }
   \IEEEauthorblockA{\IEEEauthorrefmark{1}%
                     Department of Electrical and Electronics Engineering,
                     Bilkent University,
                     arda.atalik@bilkent.edu.tr
                     }
   \IEEEauthorblockA{\IEEEauthorrefmark{2}%
                     Department of Electrical and Electronics Engineering,
                     Bogazici University,
                     alper.kose@boun.edu.tr
                     }
   \IEEEauthorblockA{\IEEEauthorrefmark{3}%
                     School of Computer and Communication Sciences, 
                     EPFL,
                     michael.gastpar@epfl.ch}
 }

\maketitle

\begin{abstract}
In the distributed remote (CEO) source coding problem, many separate encoders observe independently noisy copies
of an underlying source. The rate loss is the difference between the rate required in this distributed setting and the rate that would
be required in a setting where the encoders can fully cooperate. In this sense, the rate loss characterizes the
price of distributed processing. We survey and extend the known results on the rate loss in various settings,
with a particular emphasis on the case where the noise in the observations is Gaussian, but the underlying source is general.
\end{abstract}

\section{Introduction and Motivation}
The distributed remote (CEO) source coding problem has been presented in \cite{berger1996ceo} with the motivation to investigate the limits of a decentralized estimation task and attracted researchers' interest. In the CEO problem, the data sequence, source, cannot be observed directly and the decoder only observes the rate-limited noisy versions of the original sequence. The decoder produces an approximation of the underlying source exploiting these observations. This is illustrated in Figure~\ref{fig-CEO}.

In \cite{berger1996ceo}, Berger {\it et al.} investigate the asymptotic behavior of the minimal error frequency as the number of agents and the total data rate of encoders go to infinity where the source and observations are assumed to be discrete and memoryless. The case, in which the source is assumed to be Gaussian distributed and the observation noise is again Gaussian and the fidelity criterion being mean-square error (MSE), is referred as the quadratic Gaussian CEO problem and has been investigated in various works as \cite{viswanathan1997quadratic,oohama1998rate,prabhakaran2004rate,tavildar2009gaussian}. 
\par Eswaran and Gastpar consider the observation noise as additive Gaussian, however they allow the underlying source to be any continuous distribution with a constraint of having a finite differential entropy in \cite{eswaran_gastpar_2019}. Likewise, in this paper, the source is allowed to be arbitrarily distributed ensuring that it has a finite variance and differential entropy.

\begin{figure}
  \begin{center}
  \setlength{\unitlength}{1.5pt}
  \begin{picture}(140,130)(130,0)
    \put (127, 60) {\vector (1, 0) {13} }
    \put (130, 63) {\makebox (10,10) {$X^n$}}

    \put (142, 10) {\line (0, 1) {105} }

    \put (142, 115) {\vector (1, 0) {7} }
    \put (163, 115) {\line (1, 0) {18} }
    \put (156, 115) {\circle {10} }
    \put (153.5, 113.5) {$+$ }
    \put (156, 130) {\vector (0, -1) {8} }
    \put (158, 122) {\makebox (10,10) {$W_1^n$}}

    \put (142, 80) {\vector (1, 0) {7} }
    \put (163, 80) {\line (1, 0) {18} }
    \put (156, 80) {\circle {10} }
    \put (153.5, 78.5) {$+$ }
    \put (156, 95) {\vector (0, -1) {8} }
    \put (158, 87) {\makebox (10,10) {$W_2^n$}}

    \put (142, 10) {\vector (1, 0) {7} }
    \put (163, 10) {\line (1, 0) {18} }
    \put (156, 10) {\circle {10} }
    \put (153.5, 8.5) {$+$ }
    \put (156, 25) {\vector (0, -1) {8} }
    \put (158, 17) {\makebox (10,10) {$W_M^n$}}

    \put (170, 115) {\vector (1, 0) {17} }
    \put (175, 118) {\makebox (10,10) {$Y_1^n$}}
    \put (189, 105) {\framebox (25,20) {\sc enc 1} }
    \put (216, 115) {\line (1, 0) {3} }
    \put (221, 115) {\line (1, 0) {3} }
    \put (226, 115) {\line (1, 0) {3} }
    \put (228, 115) {\vector (1, 0) {3} }
    \put (214, 118) {\makebox (19,10) {$nR_1$}}

    \put (170, 80) {\vector (1, 0) {17} }
    \put (175, 83) {\makebox (10,10) {$Y_2^n$}}
    \put (189, 70) {\framebox (25,20) {\sc enc 2} }
    \put (216, 80) {\line (1, 0) {3} }
    \put (221, 80) {\line (1, 0) {3} }
    \put (226, 80) {\line (1, 0) {3} }
    \put (228, 80) {\vector (1, 0) {3} }
    \put (214, 83) {\makebox (19,10) {$nR_2$}}

    \put (200, 50) {\makebox (10,10) {$\vdots$}}

    \put (170, 10) {\vector (1, 0) {17} }
    \put (175, 13) {\makebox (10,10) {$Y_M^n$}}
     \put (189, 0) {\framebox (25,20) {\sc enc $M$} }
    \put (216, 10) {\line (1, 0) {3} }
    \put (221, 10) {\line (1, 0) {3} }
    \put (226, 10) {\line (1, 0) {3} }
    \put (228, 10) {\vector (1, 0) {3} }
    \put (214, 13) {\makebox (19,10) {$nR_M$}}

    \put (233, 0) {\framebox (20,125) {\sc dec} }
    \put (255, 60) {\vector (1, 0) {15} }
    \put (258, 63) {\makebox (10,10) {${\hat X}^n$}}
  \end{picture}
  \end{center}
\caption{The $M$-agent CEO problem. The {\it rate loss} is the difference in compression rate required in the distributed setting illustrated here versus the rate that would be required if all $M$ encoders could cooperate fully. It characterizes the price of distributed processing.}
\label{fig-CEO}
\end{figure}
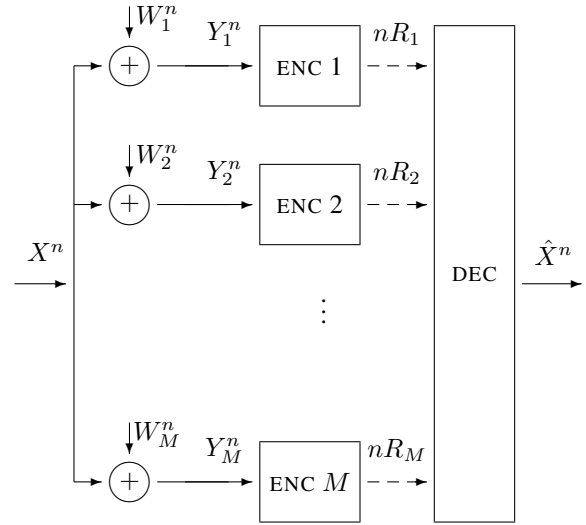
\subsection{Contribution and Outline}
In Section~\ref{exactrateloss}, we revisit the well known rate loss for the special case where the underlying source is Gaussian. We explicitly explore the limiting regimes. The main contributions of the paper are in Section~\ref{mainresultsection}:
\begin{itemize}
\item In Section~\ref{sec-remote}, we compare two known lower bounds of the rate distortion function in the remote source coding problem.
\item In Section~\ref{novelratelossbounds}, we provide novel rate loss bounds for the $M-$agent AWGN CEO problem and compare with the previous bounds analytically and numerically.
\item In Section~\ref{Sec:asymptoticrateloss}, we establish the asymptotic behavior of the rate loss bounds. For example, we consider the case where the number of agents $M$ becomes large: In this regime, the rate loss typically scales {\it linearly} in the number of agents $M.$ We also consider the case where the distortion tends to the minimum possible distortion: In this regime, the rate loss typically scales like $\log\frac{1}{\delta},$ where $\delta$ is the gap from the minimum distortion.
\item In Section~\ref{sec:comparison}, we present comparisons between the various bounds proposed in this paper as well as bounds that appear in previous work. We also show a number of numerical evaluations for concrete cases, such as when the underlying source distribution is Uniform or Laplacian.
\item Section~\ref{conclusionsection} concludes the paper and presents some future directions. 
\end{itemize}

\subsection{Notation}
We use uppercase letters $X,\,Y$ to denote random variables, and lowercase letters $x,\,y$ to denote their realizations. Given a square-integrable, absolutely continuous random variable $X$ with density $p_X(x)$, its variance $\Var{X}$ is denoted as $\sigma_X^2$, and its differential entropy is 
\begin{equation}
    h(X) = -\int p_X(x)\,\log p_X(x)\,dx.
\end{equation} 
The entropy power of $X$ is $N(X) = \dfrac{\e{2h(X)}}{2\,\pi\,e}$,\vspace*{.125cm} mutual information is  $I(X;Y) = h(Y)-h(Y\given X)=h(X)-h(X\given Y)$ and its Fisher information is 
$J(X) = \int p_X(x)\,\left(\frac{d}{dx}\,\log p_X(x)\right)^2 dx$, see~\cite[p.671]{CoverThomas06}.\vspace*{.1cm}
Denote the conditional expectation of $X$ given $Y$ as 
\begin{equation}
    V = \E{X\given Y}
\end{equation}
and its corresponding mean-square error as 
\begin{equation}
    \mmse{X\given Y} = \E{\Var{X\given Y}} = \E{(X-V)^2}.
\end{equation}
We denote the asymptotic equivalence of $f(x)$ and $g(x)$ around $x=x_0$ by $f(x)\sim g(x)$. That is, $\lim_{x\rightarrow x_0} \frac{f(x)}{g(x)}=1$.
\par In this paper, our primary focus concerns the additive Gaussian noise model, 
\begin{equation}
    Y = X+W,
    \label{themodel}
\end{equation}
where $W$ is a zero-mean Gaussian random variable of variance $\sigma_W^2,$ independent of the signal $X,$
and where $X$ has an arbitrary distribution. This is illustrated pictorially in Figure~\ref{observation}.

\begin{figure}
\begin{center}
\begin{tikzpicture}[scale=2,shorten >=1pt, auto, node distance=1cm,
   node_style/.style={scale=1,circle,draw=black,thick},
   edge_style/.style={draw=black,dashed}]

    \node [fill=none] at (-1,0) (nodeS) {$X$};
    \node [fill=none] at (0,-0.65) (nodeS) {$W \sim \mathcal{N}(0,\sigma_W^2)$};
    \node [fill=none] at (1.3,0) (nodeS) {$Y = X + W$};
    \node[node_style] (v1) at (0,0) {$+$};
   
    \draw [-stealth](-0.9,0) -- (-0.2,0);
    \draw [-stealth](0,-0.5) -- (0,-0.2);
    \draw [-stealth](0.22,0) -- (0.8,0);
    \end{tikzpicture}
\end{center}
\caption{The additive Gaussian noise observation model}
\label{observation}
\end{figure}
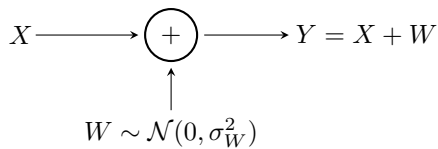

\subsection{Differential Entropy of the Conditional Mean}
The probability density function of the conditional mean is calculated explicitly in \cite{dytso_poor_shitz_2020}, but in general, it may  not be easy to calculate for arbitrary input distributions. The following identity is useful as it omits the calculation of the density of $\E{X\given Y}$.   
Our theorem is based on the following lemma, which relates the differential entropy of the conditional expectation to that of the output. 
\begin{theorem}\label{theorem-main}.
For the model given in Equation~\eqref{themodel} with $\sigma_W^2 >0,$ the differential entropy of the conditional mean can be written as 
\begin{equation}
    h(\E{X\given Y}) = h(Y) + \E{\loge{\frac{1}{\sigma_W^2}\Var{X\given Y}}}.
    \label{hV}
\end{equation}
Furthermore, we have the following lower bound:
\begin{equation}
    h(\E{X\given Y})+h(Y)\geq 2h(X). \label{mainresult}
\end{equation}
\end{theorem}
\begin{proof}
This lemma follows by careful application of several known tools, including Tweedie's formula~\cite{robbins2020empirical} and the Hatsell-Nolte identity~\cite{hatsell_nolte_1971}. A full proof is provided in \cite{atalik2021differential}. 
\end{proof}
\section{Exact Rate Loss in the Gaussian Input $M-$Agent AWGN CEO Problem}\label{exactrateloss}
The exact loss for the Gaussian input is well-known \cite{oohama1998rate}.
In this case, the smallest attainable distortion (even with unlimited rates) is well known to be
\begin{align}
    D_{\mathcal{N},0} &= \frac{\sigma_X^2\sigma_W^2}{M\,\sigma_{X}^2+\sigma_W^2}. \label{Eq-MMSE-GaussianSource}
\end{align}
The exact rate loss is given in the following proposition:
\begin{proposition}.\label{GaussianInput}
For the model given in (\ref{fig-CEO}), if the input is Gaussian with variance $\sigma_X^2$, then the exact rate loss $L_{\mathcal{N}}(D)$ for $D>D_{\mathcal{N},0}$ can be written as 
    \begin{align}
    L_\mathcal{N}(D)&=\frac{M-1}{2} \log \left(\frac{D}{\frac{\sigma_{X}^{2}+\sigma_{W}^{2} / M}{\sigma_{X}^{2}} D-\frac{\sigma_{W}^{2}}{M}}\right)\nonumber\\  
    &= \frac{M-1}{2}\log\frac{1}{1-\frac{\sigma_W^2}{M}\left(\frac{1}{D}-\frac{1}{\sigma_X^2}\right)} \label{exactGauss}
    \end{align}
\end{proposition}
From the expression, several important observations could be made on how the rate loss scales.
\subsection{Fixed $D$, large $M$}
As we expand (\ref{exactGauss}) at $M=\infty$, we obtain
\begin{equation}
  L_{\mathcal{N}}(D) = L_{\mathcal{N}}^{\infty}(D)+L_{\mathcal{N}}^{\infty}(D)\,\left(L_{\mathcal{N}}^{\infty}(D)-1\right)\frac{1}{M}+O\left(\frac{1}{M^2}\right) \label{gaussianratelossinfinity}
\end{equation}
where $L_{\mathcal{N}}^{\infty}(D) = \frac{1}{2} \sigma _W^2 \left(\frac{1}{D}-\frac{1}{\sigma _X^2}\right)$ is the rate loss as $M$ approaches infinity. Hence, for large $M$, the rate loss is inversely proportional with $D$.
\subsection{$D=\beta\,M^{-\alpha}$ for $0<\alpha\leq
1$, large $M$}
We assume $\beta$ is chosen such that $D>D_{\mathcal{N},0}$ (see Eqn.~\eqref{Eq-MMSE-GaussianSource}). In this case, the rate loss is asymptotically equivalent to $\gamma M^\alpha$. That is, 
\begin{equation}
    L_{\mathcal{N}}(D) \sim \gamma  M^\alpha \label{asympequivalence}
\end{equation} for large $M$ where 
\begin{equation}
    \gamma = \begin{cases}
    \frac{\sigma _W^2}{2\beta} & \text{ if } 0<\alpha<1\\
    \frac{1}{2} \log \left(\frac{\beta }{\beta -\sigma _W^2}\right) & \text{ if } \alpha=1, \beta\neq \sigma_W^2
    \end{cases}.
\end{equation} 
\subsection{Fixed $M$, $D=D_{\mathcal{N},0}+\delta$ for small $\delta$}
As we expand (\ref{exactGauss}) around the smallest possible distortion $D_{\mathcal{N},0}$ as given in Eqn.~\eqref{Eq-MMSE-GaussianSource}, {\it i.e.,} $\delta=0^+$, we obtain 
\begin{multline}
    L_{\mathcal{N}}(\delta) = \frac{1}{2} (M-1)\log\frac{1}{\delta}+\frac{1}{2} (M-1) \log \left(\frac{M \sigma _W^2 \sigma _X^4}{\left(M \sigma _x^2+\sigma _W^2\right){}^2}\right)\\+\frac{\delta  (M-1) \left(M \sigma _X^2+\sigma _W^2\right)}{2 \sigma _W^2 \sigma _X^2}+O\left(\delta ^2\right) \label{Gaussratelosscase3}
\end{multline}
Hence, the rate loss scales with $\log\frac{1}{\delta}$. 
In the next section, we analyze the same three cases and show that similar scaling behaviours are observed when the input is non-Gaussian. 
\section{Main Results}\label{mainresultsection}
\subsection{Lower Bounds of the Rate Distortion Function in the Remote Source Coding Problem}\label{sec-remote}
\label{mainappsection}
\begin{figure}[H]
\begin{center}
\small
\begin{tikzpicture}[scale=2,shorten >=1pt, auto, node distance=1cm,
   node_style/.style={scale=1,circle,draw=black,thick},
   edge_style/.style={draw=black,dashed}]

    \node [fill=none] at (-.4,0.15) (nodeS) {$X^n$};
    \node [fill=none] at (0,-0.65) (nodeS) {$W^n \sim \mathcal{N}(0,\sigma_W^2\,I_n)$};
    \node [fill=none] at (0.5,0.15) (nodeS) {$Y^n$};
    \node[rectangle,draw,minimum width = 1cm, 
    minimum height = .8cm] (r) at (1.1,0) {ENC};
    \node[node_style] (v1) at (0,0) {$+$};
    \node [fill=none] at (1.6,0.15) (nodeS) {$nR$};
    \node[rectangle,draw,minimum width = 1cm, 
    minimum height = .8cm] (r) at (2.15,0) {DEC};
    \node [fill=none] at (2.6,0.15) (nodeS) {$\hat{X}^n$};
   
    \draw [-stealth](-0.7,0) -- (-0.2,0);
    \draw [-stealth](0,-0.5) -- (0,-0.2);
    \draw [-stealth](0.22,0) -- (0.8,0);
    \draw [dashed,-stealth](1.4,0) -- (1.85,0);
    \draw [-stealth](2.45,0) -- (2.9,0);
    \end{tikzpicture}
\end{center}
  \caption{The AWGN remote source coding problem.}\label{Fig-AWGNremoteSC}
\end{figure}
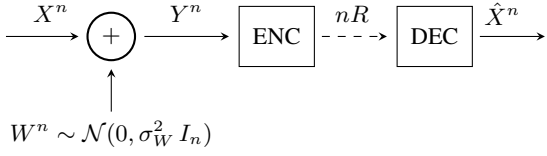

An important application of Inequality (\ref{mainresult}) can be found in the remote source coding problem.
Specifically, consider the source coding problem illustrated in Figure~\ref{Fig-AWGNremoteSC}: An encoder observes the underlying source $X$ subject to additive white Gaussian noise $W.$ The noisy observation is $Y$ and can be encoded using $R$ bits per sample. The decoder produces a reconstruction $\hat{X}$ to within the smallest possible mean-squared error. For a formal problem statement, we refer to~\cite{eswaran_gastpar_2019}.
The smallest possible rate to attain a target distortion $D$ is referred to as the {\it remote rate-distortion function,} denoted as $R_X^R(D).$
For the case where the underlying source $X$, not necessarily Gaussian, has finite differential entropy,~\cite{eswaran_gastpar_2019} discusses two different lower bounds for the remote rate-distortion function, namely
\begin{IEEEeqnarray}{rCl}
R_X^R(D) &\geq& \frac{1}{2}\log^{+} \frac{N(V)}{D} \nonumber \\
&& {}+\frac{1}{2}\log^{+} \frac{N(Y)}{N(Y)-\frac{N(X)}{D}N(W)} \label{lb1}
\end{IEEEeqnarray}
and
\begin{IEEEeqnarray}{rCl}
R_X^R(D)&\geq& \frac{1}{2}\log^{+} \frac{N(X)}{D}\nonumber\\&& {}+ \frac{1}{2}\log^{+} \frac{N(X)}{N(Y)-\frac{N(X)}{D}\sigma_W^2}
\label{lb2}
\end{IEEEeqnarray}
where $D>\E{(X-V)^2}$ and $\log^{+}x=\max\{0,\,\log x\}$. In \cite{atalik2021differential}, it is shown that (\ref{lb1}) is a tighter bound than (\ref{lb2}). Thus, in the subsequent sections, we mainly use (\ref{lb1}).  
\subsection{Novel Rate-Loss Bounds for The CEO Problem}
\label{novelratelossbounds}
In this section, we apply Theorem~\ref{theorem-main} to the so-called CEO problem.
In this problem, a single underlying source $X$ is observed by $M$ encoders. Each encoder receives a noisy version of the source $X,$ denoted as $Y_i=X+W_i,$ for $i=1, 2, \ldots, M.$ In our consideration, the noises $W_i$ are assumed to be zero-mean Gaussian, independent of each other, and of variance $\sigma_W^2.$
Each encoder compresses its observation using $R_i$ bits. All $M$ compressed representations are given to a single central decoder whose goal is to produce a reconstruction of the underlying source $X$ to with mean-squared error $D.$ The smallest possible sum-rate required to attain a distortion $D$ is denoted by $R_X^{\mathrm{CEO}}(D).$
We precisely follow the exact problem statement and notation used in~\cite{eswaran_gastpar_2019}.

The {\it rate loss} in the CEO problem denotes the difference between $R_X^{\mathrm{CEO}}(D)$ and the much smaller rate that would be required if all encoders were to cooperate fully, {\it i.e.,} the rate required by a single encoder having access to all $M$ noisy source observations.
Evidently, if the encoders are allowed to cooperate fully, then the problem is exactly the remote rate-distortion problem discussed in Section~\ref{sec-remote} above, but with reduced variance $\sigma_W^2/M$.
We denote the corresponding rate by $R_{X}^{\mathrm{R}}(D),$ and the rate loss by
\begin{equation}
    L(D) \triangleq R_X^{\mathrm{CEO}}(D)-R_{X}^{\mathrm{R}}(D).
\end{equation} 
In this section, we establish a novel bound on this rate loss.

To develop our results we will use the auxiliary notations
\begin{align}
    Y(M) &= \frac{1}{M} \sum_{i=1}^M Y_i \label{eq-defn-YM}
\end{align}
and $V(M) = \E{X\given Y(M)}.$

A lower bound for $L(D)$ is presented in~\cite{eswaran_gastpar_2019}.
For
\begin{align}
\sigma_X^2\dfrac{\sfrac{\sigma_W^2}{M}}{\sigma_X^2+\sfrac{\sigma_W^2}{M}}<D<N(X)\dfrac{\sfrac{\sigma_W^2}{M}}{N(Y(M))-N(X)},
\end{align}
the lower bound on the rate loss establishes that
\begin{multline}
    L(D)\geq \frac{M}{2}\loge {\frac{1}{\frac{N(Y(M))}{N(X)}-\frac{\sigma_W^2}{M}\frac{1}{D}}}\\
    -\frac{1}{2}\loge{\frac{\sigma_X^2}{N(X)}\,\frac{1}{1-\frac{\sigma_W^2}{M}\left(\frac{1}{D}-\frac{1}{\sigma_X^2}\right)}} \label{lowerbnd1}
\end{multline}
For $M\uparrow \infty$ and under some regularity conditions further discussed in \ref{kappaxlemma},
the bound becomes, for $0<D<\frac{1}{J(X)}$
\begin{equation}
    L(D) \geq  \frac{\sigma_W^2}{2} \left(\frac{1}{D}-{J(X)}\right)-\frac{1}{2}\log \frac{\sigma_X^2}{N(X)}.\label{Eqn-rateloss-lowerbound}
\end{equation}
The novel bound presented here is an {\it upper} bound on the rate loss, developed in the following subsections.

\subsubsection{Cooperation Bound}
The first ingredient of the novel upper bound on the rate loss is an improved lower bound on $R_{X}^{\mathrm{R}}(D).$
 To this end, we will utilize both $N(V(M))$ and $\mmse{X\given Y(M) }$, i.e., for all $D> \mmse{X\given Y(M)}$
\begin{align}
    R_{X}^{\mathrm{R}}(D) &\geq \frac{1}{2}\log^+ \frac{N(V(M))}{D-\mmse{X\given Y(M)}} \label{tight}
\end{align}

One can weaken (\ref{tight}) to omit the calculation of $\mmse{X\given Y(M)}$. In that case, one obtains for all $D>N(X)\,\sigma_W^2/\left(M\,N(Y(M))\right)$,
\begin{IEEEeqnarray}{rCl}
    R_{X}^{\mathrm{R}}(D) &\geq& \frac{1}{2}\log^{+}\frac{N(V(M))}{D}\\
    && {}+\frac{1}{2}\log^{+}\frac{M\, N\left(Y(M)\right)}{M\, N\left(Y(M)\right)-\frac{N(X)}{D}N(W)}
\end{IEEEeqnarray}
where $Y(M)$ was defined in Equation~\eqref{eq-defn-YM}, and $V(M) = \E{X\given Y(M)}$. As we have shown in Section~\ref{sec-remote}, this bound is tighter than the other lower bound in \cite{eswaran_gastpar_2019} for any finite\footnote{Observe that as $M\uparrow \infty$, the second term vanishes, and the bound becomes $R_{X}^{\mathrm{R}}(D)\geq \frac{1}{2}\log^+ \frac{N(X)}{D}$ as expected. This is also true for the other lower bound.} $M$.  
\subsubsection{Novel Rate Loss Upper Bound}
In order to upper bound the rate loss $L(D),$ we utilize the upper bound on the CEO sum-rate distortion by Eswaran and Gastpar \cite{eswaran_gastpar_2019}, which states that for $D>D_{\mathcal{N},0}$, 
\begin{IEEEeqnarray}{rCl}
    R_X^{\mathrm{CEO}}(D) &\leq& \frac{1}{2}\log^+\frac{\sigma_X^2}{D}+\frac{M}{2}\log^+ \frac{M\,\sigma_X^2}{M \sigma_{Y(M)}^2-\frac{\sigma_X^2}{D}\sigma_W^2}\\
    &=& \frac{1}{2}\log^+\frac{\sigma_X^2}{D}+\frac{M}{2}\log^+\frac{1}{1+\frac{\sigma_W^2}{M}\left(\frac{1}{\sigma_X^2}-\frac{1}{D}\right)}\\
    &=& \begin{cases}\frac{1}{2}\log\frac{\sigma_X^2}{D}+\frac{M}{2}\log\frac{1}{1+\frac{\sigma_W^2}{M}\left(\frac{1}{\sigma_X^2}-\frac{1}{D}\right)} & \small\text{if } \scriptstyle{D<\sigma_X^2}\\
    0 & \small\text{otherwise} \label{rceo}
    \end{cases}.
\end{IEEEeqnarray}

One can use (\ref{tight}) to obtain a tight upper bound on the rate loss. 
\begin{theorem}.\label{thmtight}
For $D_{\mathcal{N},0}<D<\mmse{X\given Y(M)}+N(V(M))$, the rate loss is upper bounded as 
\begin{IEEEeqnarray}{rCl}
    L(D) &\leq& \frac{1}{2}\log \frac{\sigma_X^2}{N(V(M))}+\frac{M}{2}\log\frac{1}{1-\frac{\sigma_W^2}{M}\left(\frac{1}{D}-\frac{1}{\sigma_X^2}\right)}\nonumber\\
     && {}-\frac{1}{2}\log \left(\frac{D}{D-\mmse{X\given Y(M)}}\right)\nonumber\\&=&
     {}\frac{M}{M-1}L_{\mathcal{N}}(D)+\frac{1}{2}\log \frac{\sigma_X^2}{N(V(M))}\nonumber\\&&{}-\frac{1}{2}\log \left(\frac{D}{D-\mmse{X\given Y(M)}}\right)
     \label{tightub}
\end{IEEEeqnarray}
\end{theorem}
\begin{proof}
Follows directly from subtracting (\ref{tight}) from (\ref{rceo}). For the regions, observe that $\mmse{X\given Y(M)}+N(V(M))\leq \mmse{X\given Y(M)}+\sigma_{V(M)}^2=\sigma_X^2$ by law of total variance.
\end{proof}
\begin{corollary}.\label{cor1}
As $M\uparrow \infty$, the upper bound on the loss becomes
\small
\begin{align}
    L(D) \leq
    \frac{1}{2}\log\frac{\sigma_X^2}{N(X)}+\frac{1}{2} \sigma _W^2 \left(\frac{1}{D}-\frac{1}{\sigma _X^2}\right) \text{ for } 0<D<N(X)\label{newBound}
\end{align}
\end{corollary}

\begin{remark}.
    Note that (\ref{tightub}) is minimized for Gaussian inputs since both $N(V(M))$ and $\mmse{X\given Y(M)}$ is maximized in that case. Furthermore, the theorem simplifies to 
    \begin{IEEEeqnarray}{rCl}
        L(D)&\leq& \frac{1}{2} \log \left(1-\frac{\sigma_W^2}{M}\left(\frac{1}{D}-\frac{1}{\sigma_X^2}\right)\right)\nonumber\\
        && {}+\frac{M}{2}\log\frac{1}{1-\frac{\sigma_W^2}{M}\left(\frac{1}{D}-\frac{1}{\sigma_X^2}\right)}\\
        &=& \frac{M-1}{2}\log\frac{1}{1-\frac{\sigma_W^2}{M}\left(\frac{1}{D}-\frac{1}{\sigma_X^2}\right)}
    \end{IEEEeqnarray}
    for any $D_{\mathcal{N},0}<D<\sigma_X^2$. 
Hence, the new upper bound is tight for Gaussian inputs, irrespective of $M$.  
\end{remark}
\subsection{Asymptotic Analysis of the Bounds}\label{Sec:asymptoticrateloss}
In this subsection, we provide an analysis similar to Section~\ref{exactrateloss}. 

\subsubsection{Fixed $D$, large $M$}

\begin{theorem}. As the number of agents increases, (\ref{tightub}) simplifies to the following. 
\begin{IEEEeqnarray}{rCl}
   L(D)&\leq& \frac{1}{2} \sigma_W^2 \left(\frac{1}{D}-\frac{1}{\sigma _X^2}\right)+\frac{1}{2}\log\frac{\sigma_X^2}{N(X)}+ \nonumber\\
     && \hspace*{-.35cm}{}\frac{\left(\frac{\sigma_W^2}{2}\left(\frac{1}{D}-\frac{1}{\sigma_X^2}\right)\right)^2-\frac{\sigma_W^2}{2}\left(\frac{1}{D}-J(X)\right)}{M}+O\left(\frac{1}{M^2}\right) \nonumber\\
     &=& L_{\mathcal{N}}^\infty(D)+\frac{1}{2}\log\frac{\sigma_X^2}{N(X)}\nonumber\\&&{}+\frac{\left(L_{\mathcal{N}}^\infty(D)\right)^2-L_{\mathcal{N}}^\infty(D)+\frac{\sigma_{W}^2}{2}\left(J(X)-\frac{1}{\sigma_X^2}\right)}{M}\nonumber\\&&{}+O\left(\frac{1}{M^2}\right)
\end{IEEEeqnarray}
and (\ref{lowerbnd1}) simplifies to 
\begin{IEEEeqnarray}{rCl}
   L(D)&\geq& \frac{1}{2} \sigma_W^2 \left(\frac{1}{D}-J(X)\right)-\frac{1}{2}\log\frac{\sigma_X^2}{N(X)}+ \nonumber\\
     && \frac{\left(\frac{\sigma_W^2}{2}\left(\frac{1}{D}-J(X)\right)\right)^2-\frac{\sigma_W^2}{2}\left(\frac{1}{D}-\frac{1}{\sigma_X^2}\right)}{M}+O\left(\frac{1}{M^2}\right) \nonumber\\&=&{} L_{\mathcal{N}}^\infty(D)-\frac{\sigma_W^2}{2}\left(J(X)-\frac{1}{\sigma_X^2}\right)-\frac{1}{2}\log\frac{\sigma_X^2}{N(X)}+\nonumber\\&& \frac{\left(\frac{\sigma_W^2}{2}\left(\frac{1}{D}-J(X)\right)\right)^2-\frac{\sigma_W^2}{2}\left(\frac{1}{D}-\frac{1}{\sigma_X^2}\right)}{M}+O\left(\frac{1}{M^2}\right)\nonumber\\&&
\end{IEEEeqnarray}
\end{theorem}
\begin{proof}
Both inequalities follow from relaxing the bounds and expanding them at $M=\infty$. By (\ref{mainresult}), $N(V(M))\geq \frac{N^2(X)}{N(Y(M))}$. We refer to Eqn. 16 in \cite{eswaran_gastpar_2019} for a lower bound on $\mmse{X\given Y(M)}$, Eqn. 98 for an upper bound of $N(Y(M))$, and Appendix~\ref{kappacalc} for the simplification of the bound. 
\end{proof}

\subsubsection{$D=\beta\,M^{-\alpha}$ for $0<\alpha\leq 1$, large $M$}
In this case, non-Gaussian inputs are not much different than Gaussian. That is, the same asymptotic equivalence (\ref{asympequivalence}) is observed. 
\begin{theorem}.
For arbitrary inputs with finite variance and  entropy power, the rate loss is asymptotically equivalent to $\gamma M^\alpha$, i.e., 
\begin{equation}
    L(D) \sim \gamma  M^\alpha \label{asympequivalence2}
\end{equation} for large $M$ where 
\begin{equation}
    \gamma = \begin{cases}
    \frac{\sigma _W^2}{2\beta} & \text{ if } 0<\alpha<1,\\
    \frac{1}{2} \log \left(\frac{\beta }{\beta -\sigma _W^2}\right) & \text{ if } \alpha=1, \beta\neq \sigma_W^2.
    \end{cases}
\end{equation}
\end{theorem}

\begin{proof}
The dominating term on the right hand side of (\ref{lowerbnd1}) is $-\frac{M}{2}\loge {{\frac{N(Y(M))}{N(X)}-\frac{\sigma_W^2}{M}\frac{M^\alpha}{\beta}}}$. By Eqn. 98 in \cite{eswaran_gastpar_2019}, this can be further relaxed into $\frac{M}{2}\loge{1-\frac{\sigma_W^2}{\beta}M^{\alpha-1}+\sigma_W^2 J(X)M^{-1}}$. Similarly, the dominating term on the right hand side of (\ref{tightub}) is $\frac{M}{2}\loge{1-\frac{\sigma_W^2}{\beta}M^{\alpha-1}+\frac{\sigma_W^2}{\sigma_X^2}M^{-1}}$. Taking the limit as $M\uparrow\infty$ of both the upper and lower bounds conclude the proof.
\end{proof}

\subsubsection{Fixed $M$, $D=D_{\mathcal{N},0}+\delta$ for small $\delta$}
For this case, the upper bound for non-Gaussian inputs exhibits the same behavior as Gaussian inputs.
\begin{theorem}. For arbitrary inputs with finite variance and  entropy power, as $D$ gets closer to  $D_{\mathcal{N},0}$ by $\delta$,
\begin{align}
    L(D)\leq g(\delta) 
\end{align}
where $g(\delta)\sim \frac{M}{2}\log\frac{1}{\delta}$. That is, \begin{equation}
    \lim_{\delta\downarrow 0} \frac{g(\delta)}{\frac{M}{2}\log{\frac{1}{\delta}}} = 1
\end{equation}
\end{theorem}
\begin{proof}
Follows immediately by (\ref{Gaussratelosscase3}) and (\ref{tightub}) as the last two terms on (\ref{tightub}) are $O(1)$. 
\end{proof}
\subsection{Comparison and Numerical Results}\label{sec:comparison}
We also note that the following upper bound on the rate loss $L(D)$ appears in~\cite{dragotti2009distributed}. 
\begin{IEEEeqnarray}{rCl}
L(D) &\leq&\frac{M-1}{2} \log \frac{1}{1-\frac{\sigma_{W}^{2}}{M}\left(\frac{1}{D}-\frac{1}{\sigma_X^2}\right)}\nonumber\\
&& {}+\frac{1}{2} \loge{\scriptstyle 1+\left(\frac{1}{D}-\frac{1}{\sigma_X^2}\right)\left(\frac{D+2\sqrt{D}\sigma_W+\frac{\sigma_W^2}{M}}{1-\frac{\sigma_W^2}{M}\left(\frac{1}{D}-\frac{1}{\sigma_X^2}\right)}\right)} \label{prevupperbound}
\end{IEEEeqnarray}

and as $M\uparrow \infty$, we have 
\begin{multline}
L(D) \leq \frac{\sigma_{W}^{2}}{2}\left(\frac{1}{D}-\frac{1}{\sigma_{X}^{2}}\right)\\
+\frac{1}{2} \log \left(1+\left(\frac{1}{D}-\frac{1}{\sigma_{X}^{2}}\right)\left(D+2 \sqrt{D \sigma_{W}^{2}}\right)\right). \label{prevBound}
\end{multline}

Comparing (\ref{tightub}) and (\ref{prevupperbound}) for any input distributions is tedious. For Gaussian inputs, it is easy to see that the new bound achieves the exact value while (\ref{prevupperbound}) does not. 
\par For comparing the two bounds in the large $M$ regime, we set $\sigma_X^2=1, \sigma_W^2=1/\text{snr}$ and solve for $D$ such that (\ref{prevBound}) is greater than  (\ref{newBound}), i.e., $D$ for which the new bound is strictly better. These regions are in the form $D<D^{*}$, and we plot $D^{*}$ vs $N(X)$ in Fig.~\ref{fig:comparison} for different snr values.  
\begin{figure}[htb]
    \centering
    \includegraphics[width = \linewidth]{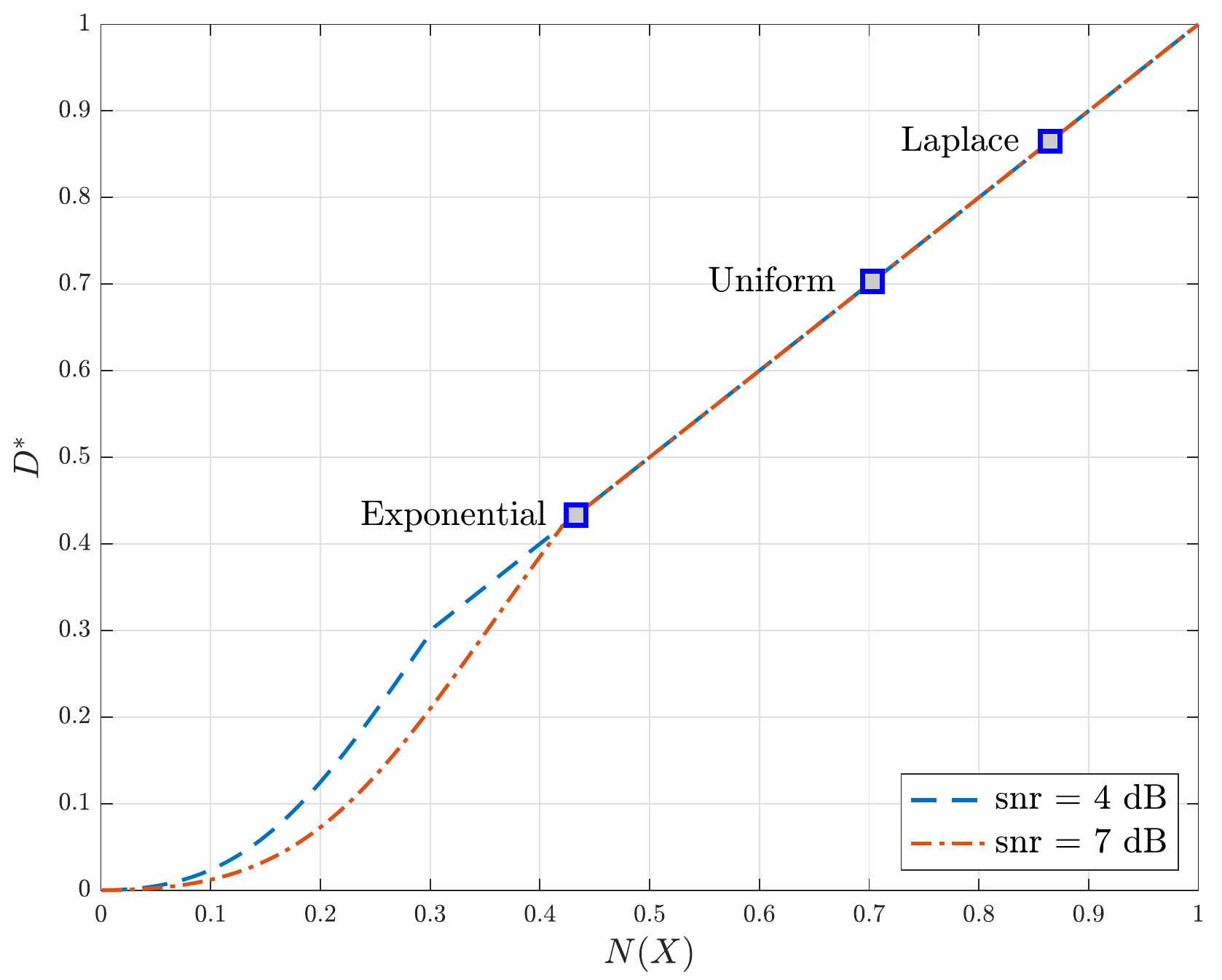}
    \caption{Comparison of the bounds (\ref{newBound}) and (\ref{prevBound}): $D^*$ vs $N(X).$ The new bound is valid for $0<D<N(X)$ and is better whenever $D<D^*.$ Note that $D^*$ is calculated numerically.}
    \label{fig:comparison}
\end{figure}
For other distributions such as Laplace and Uniform, we present numerical results in Fig.~\ref{fig1}, and Fig.~\ref{fig2} for $M\uparrow \infty$, and in Fig.~\ref{fig3} for $M=10$. 

\begin{figure}[htb]
    \centering
    \includegraphics[width = \linewidth]{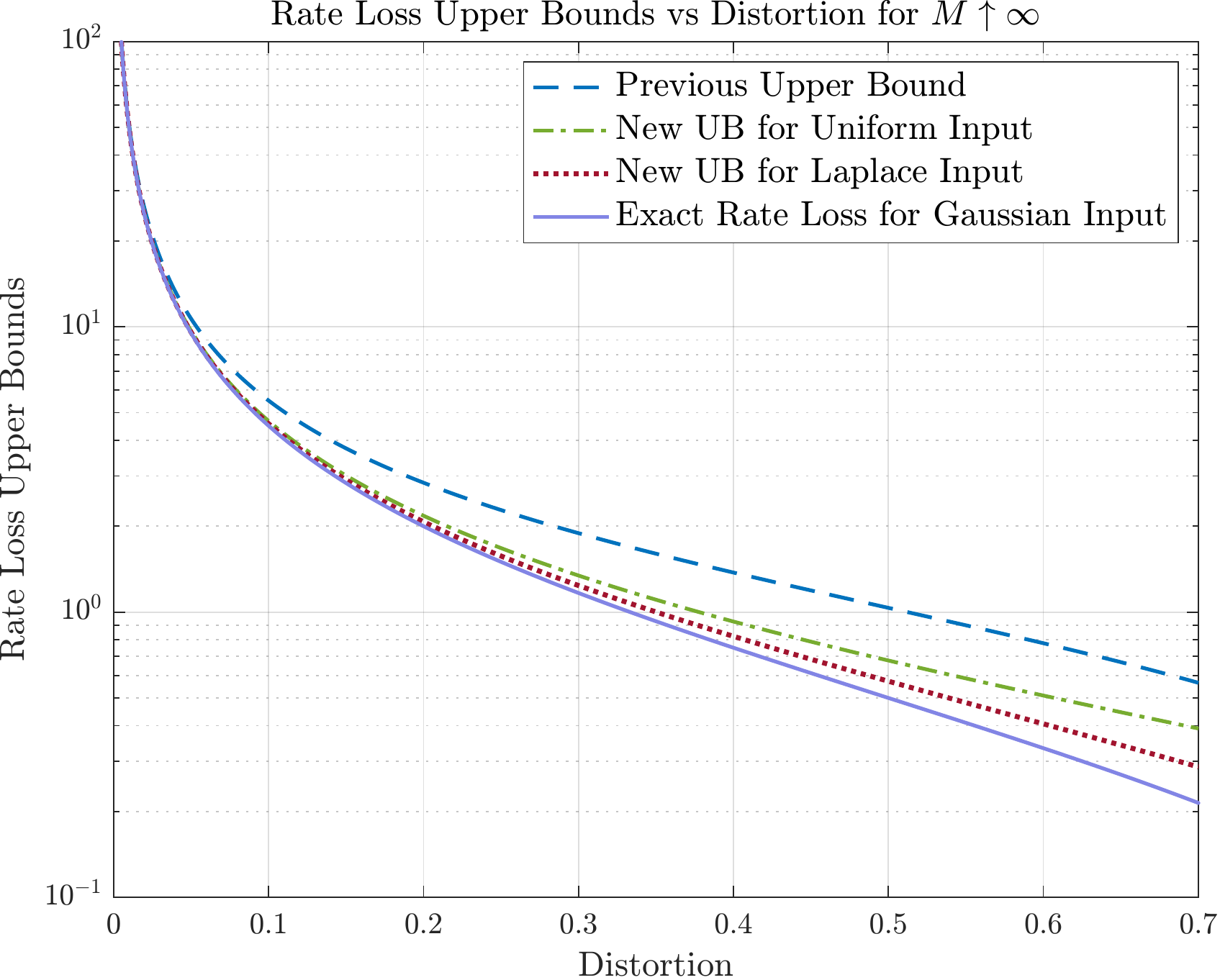}
    \caption{Comparison of Upper Bounds for $\sigma_X^2=\sigma_W^2=1, M\uparrow\infty$: Previous upper bound refers to (\ref{prevBound}), new upper bound refers to (\ref{newBound}), and exact rate loss for Gaussian input refers to $L_{\mathcal{N}}^\infty(D)$ in (\ref{gaussianratelossinfinity}). Calculations are done analytically.}
    \label{fig1}
\end{figure}

\begin{figure}[htb]
    \centering
    \includegraphics[width = \linewidth]{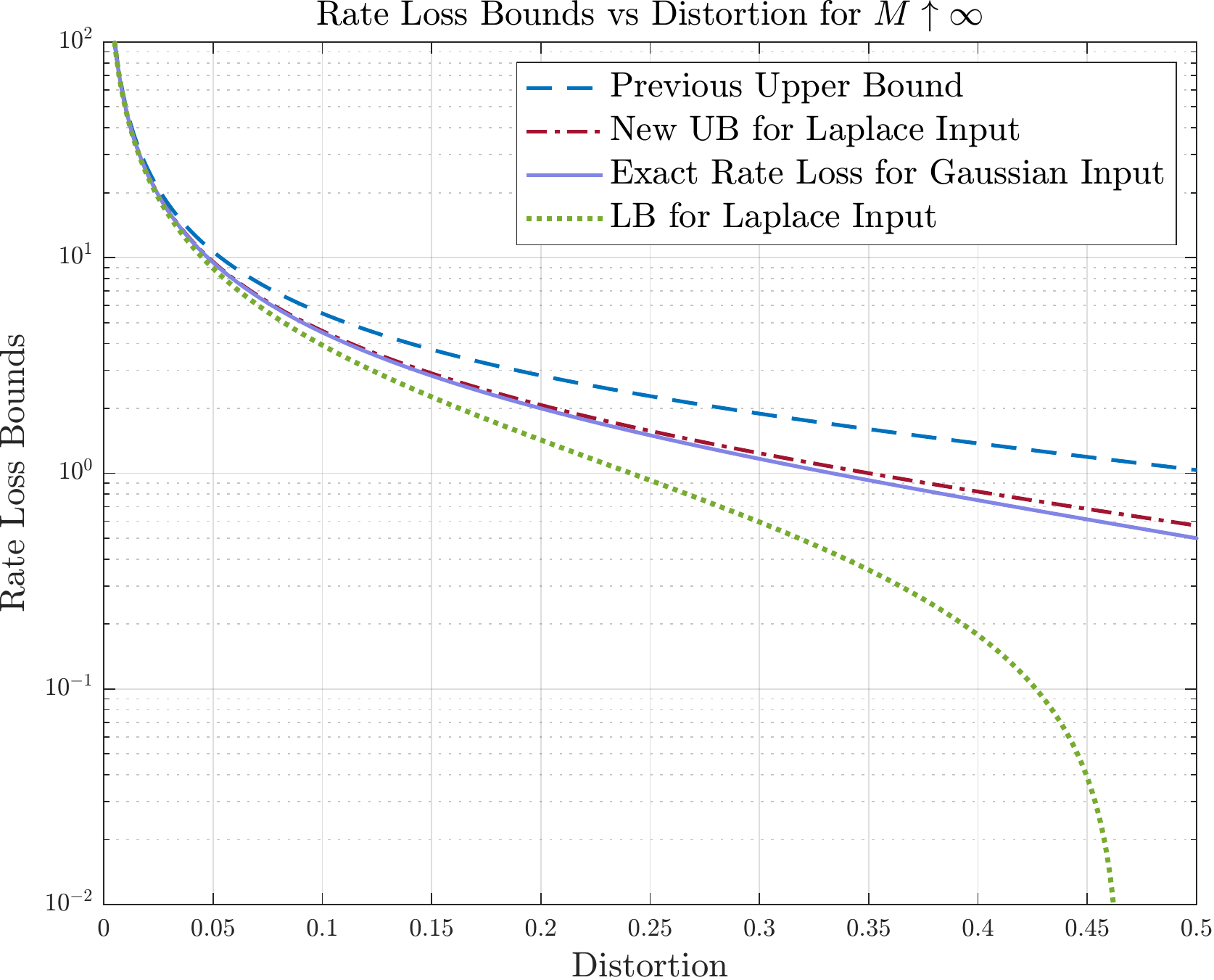}
    \caption{Comparison of Bounds for $\sigma_X^2=\sigma_W^2=1, M\uparrow\infty$: Previous upper bound refers to (\ref{prevBound}), new upper bound refers to (\ref{newBound}), exact rate loss for Gaussian input refers to $L_{\mathcal{N}}^\infty(D)$ in (\ref{gaussianratelossinfinity}), and lower bound refers to (\ref{Eqn-rateloss-lowerbound}). Calculations are done analytically.}
    \label{fig2}
\end{figure}

\begin{figure}[htb]
    \centering
    \includegraphics[width = \linewidth]{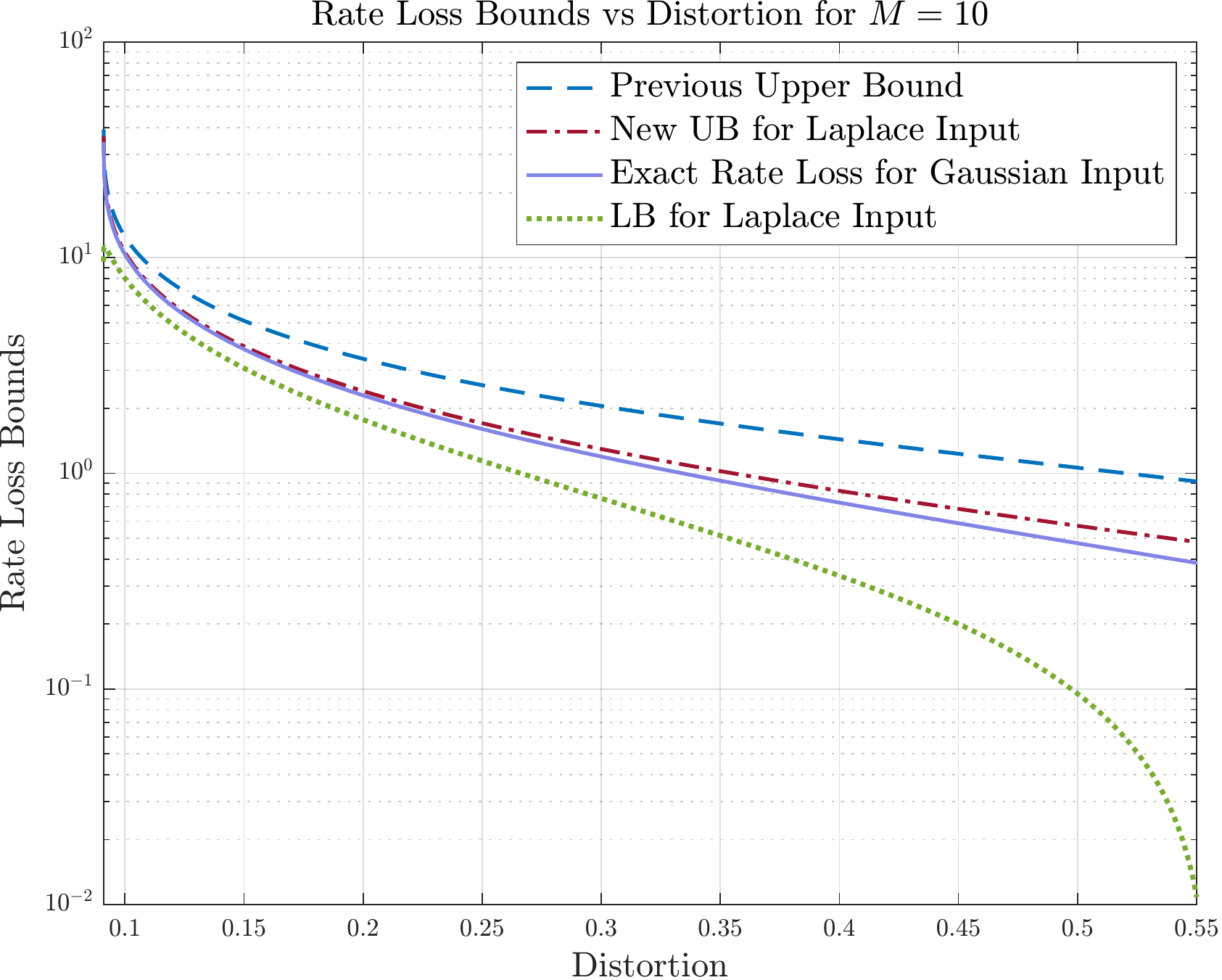}
    \caption{Comparison of Bounds for $\sigma_X^2=\sigma_W^2=1$, $M=10$: Previous upper bound refers to (\ref{prevupperbound}), new upper bound refers to (\ref{tightub}), exact rate loss for Gaussian input refers to (\ref{exactGauss}), and lower bound refers to (\ref{lowerbnd1}). Calculations are done numerically.}
    \label{fig3}
\end{figure}
\begin{remark}:
It is important to note that the Gaussian input maximizes the lower bound (\ref{lowerbnd1}), whereas it minimizes the upper bound (\ref{tightub}). Hence, the bounds are tight for the inputs that are close to the Gaussian distribution in terms of KL-divergence.  \end{remark}

\section{Conclusion and Outlook}\label{conclusionsection}
We studied the well known rate loss where the underlying source is arbitrary but having a finite variance and differential entropy. We explicitly explored three limiting regimes: the number of users gets larger for fixed $D$, the number of users gets larger and $D$ is approaching zero with $1/D$ staying sub-linear, and small $D$ for fixed number of users. Our results indicate that the arbitrary input case is not so different from the Gaussian input analogue which is due to \cite{oohama1998rate}. An interesting continuation would be the study of the worst-case rate loss and its tight bounds.

 \section*{Acknowledgments}
 The work in this manuscript was supported in part by the Swiss National Science Foundation under Grant 200364.

\newpage

\begin{appendix}
\label{kappacalc}
\begin{lemma}:  Under regularity conditions,
\begin{align}
    \kappa_X &\triangleq \lim_{s\rightarrow 0^+}\frac{d}{ds} N(X+\sqrt{s}G)\\
    &  = N(X)\, J(X). \label{kappaxlemma}
\end{align}  
\end{lemma}
\newpage
\begin{proof}
We start by observing that we can express
\begin{align}
    \kappa_X &= \lim_{s\rightarrow 0^+}\frac{d}{ds} N(X+\sqrt{s}G)\\
    &= \lim_{s\rightarrow 0^+} 2N(X+\sqrt{s}G)\,\frac{d}{ds} h(X+\sqrt{s}G)\\
    &= 2N(X)\lim_{s\rightarrow 0^+} \frac{d}{ds} h(X+\sqrt{s}G)\\
    &= 2N(X) \,\lim_{s\rightarrow 0^+}\frac{1}{2}\, J(X+\sqrt{s}G)\\
    &= N(X)\, J(X) ,
\end{align}
where we use de Bruijn’s Identity ~\cite[p.672]{CoverThomas06} in the last two lines. 
\end{proof}
\end{appendix}
\bibliographystyle{IEEEtran}
\bibliography{biblio}

\begin{thebibliography}{10}
\providecommand{\url}[1]{#1}
\csname url@samestyle\endcsname
\providecommand{\newblock}{\relax}
\providecommand{\bibinfo}[2]{#2}
\providecommand{\BIBentrySTDinterwordspacing}{\spaceskip=0pt\relax}
\providecommand{\BIBentryALTinterwordstretchfactor}{4}
\providecommand{\BIBentryALTinterwordspacing}{\spaceskip=\fontdimen2\font plus
\BIBentryALTinterwordstretchfactor\fontdimen3\font minus
  \fontdimen4\font\relax}
\providecommand{\BIBforeignlanguage}[2]{{%
\expandafter\ifx\csname l@#1\endcsname\relax
\typeout{** WARNING: IEEEtran.bst: No hyphenation pattern has been}%
\typeout{** loaded for the language `#1'. Using the pattern for}%
\typeout{** the default language instead.}%
\else
\language=\csname l@#1\endcsname
\fi
#2}}
\providecommand{\BIBdecl}{\relax}
\BIBdecl

\bibitem{berger1996ceo}
T.~Berger, Z.~Zhang, and H.~Viswanathan, ``The ceo problem [multiterminal
  source coding],'' \emph{IEEE Transactions on Information Theory}, vol.~42,
  no.~3, pp. 887--902, 1996.

\bibitem{viswanathan1997quadratic}
H.~Viswanathan and T.~Berger, ``The quadratic gaussian ceo problem,''
  \emph{IEEE Transactions on Information Theory}, vol.~43, no.~5, pp.
  1549--1559, 1997.

\bibitem{oohama1998rate}
Y.~Oohama, ``The rate-distortion function for the quadratic gaussian ceo
  problem,'' \emph{IEEE Transactions on Information Theory}, vol.~44, no.~3,
  pp. 1057--1070, 1998.

\bibitem{prabhakaran2004rate}
V.~Prabhakaran, D.~Tse, and K.~Ramachandran, ``Rate region of the quadratic
  gaussian ceo problem,'' in \emph{International Symposium onInformation
  Theory, 2004. ISIT 2004. Proceedings.}\hskip 1em plus 0.5em minus 0.4em\relax
  IEEE, 2004, p. 119.

\bibitem{tavildar2009gaussian}
S.~Tavildar, P.~Viswanath, and A.~B. Wagner, ``The gaussian many-help-one
  distributed source coding problem,'' \emph{IEEE Transactions on Information
  Theory}, vol.~56, no.~1, pp. 564--581, 2009.

\bibitem{eswaran_gastpar_2019}
K.~Eswaran and M.~Gastpar, ``Remote source coding under gaussian noise: Dueling
  roles of power and entropy power,'' \emph{IEEE Trans. Inf. Theory}, vol.~65,
  no.~7, pp. 4486--4498, Jul. 2019.

\bibitem{CoverThomas06}
T.~Cover and J.~Thomas, \emph{Elements of Information Theory}, 2nd~ed.\hskip
  1em plus 0.5em minus 0.4em\relax John Wiley and Sons, 2006.

\bibitem{dytso_poor_shitz_2020}
A.~Dytso, H.~V. Poor, and S.~Shamai, ``A general derivative identity for the
  conditional mean estimator in gaussian noise and some applications,'' in
  \emph{Proc. IEEE Int. Symp. Inf. Theory}, Los Angeles, CA, USA, Jun. 2020,
  pp. 1183--1188.

\bibitem{robbins2020empirical}
H.~Robbins, ``An empirical {B}ayes approach to statistics,'' in \emph{Proc.
  Third Berkeley Symp. Math Statist. Probab.}, vol.~1, 1956, pp. 157--163.

\bibitem{hatsell_nolte_1971}
C.~Hatsell and L.~Nolte, ``Some geometric properties of the likelihood ratio
  (corresp.),'' \emph{IEEE Trans. Inf. Theory}, vol.~17, no.~5, pp. 616--618,
  1971.

\bibitem{atalik2021differential}
A.~Atalik, A.~Köse, and M.~Gastpar, ``Differential entropy of the conditional
  expectation under gaussian noise,'' 2021, arxiv:2106.04677.

\bibitem{dragotti2009distributed}
P.~L. Dragotti and M.~Gastpar, \emph{Distributed source coding: theory,
  algorithms and applications}.\hskip 1em plus 0.5em minus 0.4em\relax Academic
  Press, 2009.

\end{thebibliography}


\end{document}